\documentclass{sig-alternate}
\CopyrightYear{2012} % Allows default copyright year (20XX) to be over-ridden - IF NEED BE.
\title{Computing optimal k-regret minimizing sets with top-k depth contours}

\numberofauthors{1}
\author{
\alignauthor Sean Chester, Alex Thomo, S. Venkatesh, and Sue Whitesides\\
       \affaddr{Computer Science Department}\\
       \affaddr{University of Victoria}\\
       \affaddr{PO Box 1700 STN CSC}\\
       \affaddr{Victoria, Canada}\\
       \email{\{schester, sue\}@uvic.ca, \{thomo, venkat\}@cs.uvic.ca}
}
\date{\today}

\usepackage{hyperref}
\hypersetup{
 pdfauthor = {Chester et al.},
 pdftitle = {Computing optimal k-regret minimizing sets with top-k depth contours},
 pdfsubject = {regret minimization},
 pdfkeywords = {regret, representative databases, top-k, arrangement of lines, plane sweep},
 colorlinks=false
}
\usepackage{balance}
\usepackage[nocompress]{cite}
\usepackage{times,epsfig}
\usepackage{amssymb,amsmath,amsfonts}%,amscd}
\usepackage{graphicx}
\usepackage{color}
\usepackage{url}

\usepackage{subfigure}
\usepackage{algorithm, algorithmic, float}
\usepackage{multirow}

\usepackage{color}
\definecolor{grayish}{gray}{0.5}

\newtheorem{theorem}{Theorem}
\newtheorem{lemma}{Lemma}[section]

\newtheorem{corollary}[lemma]{Corollary}

\newtheorem{definition}{Definition}[section]

\newcommand{\ea}{~et~al.}
\newcommand{\nanon}{Nanongkai\ea}

\newcommand{\ds}{D} %dataset D
\renewcommand{\ss}{S} %subset of D, S
\newcommand{\famf}{\mathcal{F}} % family of functions F
\newcommand{\uf}{f} % utility function f
\newcommand{\point}{p} % point p in D

\newcommand{\gain}{gain} % gain
\newcommand{\regret}{r_{\ds}} % regret
\newcommand{\rratio}{rr_{\ds}} % regret ratio
\newcommand{\regretsize}{c} % # of lines in a regret minimizing set
\newcommand{\kgain}{k\gain}
\newcommand{\kregret}{k\regret}
\newcommand{\krratio}{k\rratio}
\newcommand{\supremum}{\mathrm{sup}} % supremum
\newcommand{\mins}{\mathcal{S}} % regret minimizing set
\newcommand{\faml}{\mathcal{L}^+} % family of linear functions
\newcommand{\famUnit}{\mathcal{U}^+} % family of positive linear functions such that sqrt(x_1^2 + ... x_n^2)=1
\newcommand{\env}{\mathcal{E}} % Envelope
\newcommand{\cchain}{\mathcal{C}} % a convex chain

\newcommand{\ls}{\mathcal{L}} % set of lines
\newcommand{\as}{\mathcal{P}} % answer set, n*m array of best solutions so far.
\newcommand{\queue}{\mathcal{Q}} % priority queue, Q
\newcommand{\pq}{\queue} % Alias for priority queue, Q
\newcommand{\disttock}{\Delta} % Distance of a point p to the contour C_k

\newcommand{\dd}{top-$k$ rank depth}

\newcommand{\origin}{O}

\newcommand{\contour}{\mathcal{C}} % contour

\newcommand{\NP}{\textsl{NP}}
\newcommand{\NPH}{{\NP}-Hard}

\renewcommand{\max}{\mathrm{max}}

\newcommand{\order}{\mathcal{O}}

\graphicspath{{figures/}}

\begin{document}
\maketitle

\begin{abstract}
Regret minimizing sets are a very recent approach to representing a dataset $\ds$ with a small subset $\ss$ of representative tuples.  The set $\ss$ is chosen such that executing any top-$1$ query on $\ss$ rather than $\ds$ is minimally perceptible to any user.
To discover an optimal regret minimizing set of a predetermined cardinality is conjectured to be a hard problem.  In this paper, we generalize the problem to that of finding an optimal $k$-regret minimizing set, wherein the difference is computed over top-$k$ queries, rather than top-$1$ queries.

We adapt known geometric ideas of top-$k$ depth contours and the reverse top-$k$ problem.  We show that the depth contours themselves offer a means of comparing the optimality of regret minimizing sets with $L_2$ distance.  We design an $\order({\regretsize}n^2)$ plane sweep algorithm for two dimensions to compute an optimal regret minimizing set of cardinality $\regretsize$.  For higher dimensions, we introduce a greedy algorithm that progresses towards increasingly optimal solutions by exploiting the transitivity of $L_2$ distance.

\end{abstract}

%A category including the fourth, optional field follows...

%, Experimentation}


\section{Introduction}\label{sec:introduction}

\begin{table}[bth]
 \centering
  \begin{tabular}{c l | c c c c }
id&player name&points&rebs&steals&fouls\\\hline
1&Kevin Durant		&{\bf 2472}	&623		 &112		 &171\\
2&LeBron James		&2258		&554		 &125		 &119\\
3&Dwyane Wade		&2045		&373		 &{\bf 142}	 &181\\
4&Dirk Nowitzki		&2027		&520		 &70		 &208\\
5&Kobe Bryant		&1970		&391		 &113 		 &187\\
6&Carmelo Anthony	&1943		&454		 &88  		 &225\\
7&Amare Stoudemire	&1896		&732	 	 &52  		 &{\bf 281}\\
8&Zach Randolph		&1681		&{\bf 950}	 &80  		 &226\\
\hline
  \end{tabular}
  \caption{\label{fig:nba}Statistics for the top eight NBA point scorers from the 2009 regular season, taken from \url{databasebasketball.com}.  The top score in each statistic is bolded.}
\end{table}

For a user navigating a large dataset, the availability of a succinct representative subset of the data is crucial.  For example, consider Table~\ref{fig:nba}, a toy, but real, dataset consisting of the top eight scoring NBA players from the $2009$ basketball season.  A user viewing this data would typically be curious which of these eight players were ``top of the class'' that season.  That is, he is curious which few tuples best represent the entire dataset, without his having to peruse it in entirety.

A well-established approach to representing a dataset is with the {\em skyline} operator which returns all pareto-optimal points.\footnote{Pareto-optimal points are those for which no other point is higher ranked with respect to every attribute.}  The intention of the skyline operator is to reduce the dataset down to only those tuples that are guaranteed to best suit the preferences or interests of {\em somebody}.  If the toy dataset in Table~\ref{fig:nba} consisted only of the attributes \textit{points} and \textit{rebounds}, then the skyline would consist only of the players \textit{Kevin Durant}, \textit{Amare Stoudemire}, and \textit{Zach Randolph}, so these three players would represent well what are the most impressive combinations of point-scoring and rebounding statistics.  The skyline is a powerful summary operator only on low dimensional datasets, however; even on this toy example, {\em everybody} is in the skyline if we consider all four attributes.  In general, there is no guarantee that the skyline is an especially succinct representation of the dataset.

\subsubsection*{Regret}
A promising new alternative is the {\em regret minimizing set}, introduced by \nanon~\cite{regretMin}, which hybridizes the skyline operator with top-$k$ queries.  A top-$k$ query takes as input a utility function $\uf$ and evaluates each tuple according to $\uf$, reporting the $k$ tuples with highest values.  Figure~\ref{fig:ranking_nba} shows how highly three of the points rank for a user utility function of $f(\mathrm{pts},\mathrm{rebs})=(\mathrm{pts}+\mathrm{rebs})/2$, if the attributes are normalized.  The distance from the orthogonal line of each point is proportional to the point's score for that user weight.  This reveals that Randolph earns the highest normalized score ($0.840$), compared to Kevin Durant ($0.828$) and then Kobe Bryant ($0.604$).  

To evaluate whether a subset effectively represents the entire dataset well, \nanon\ introduce {\em regret ratio} as the ratio of how far from the best score in the dataset is the best score in that subset, with respect to a given utility function.  Graphically, this is proportional to how much smaller than the largest arrow is the largest arrow in the subset.  For the subset $\{\mathrm{Bryant}, \mathrm{Durant}\}$, the regret ratio is:
$$(0.840-0.828)/0.840=0.0143,$$
since the score for Randolph is the best in the dataset at 0.840, and the score for Durant is the best in the subset at 0.828.

\begin{figure}[tbh]
 \centering
 \includegraphics[scale=.35, clip=true, trim=0 14.5cm 0 0]{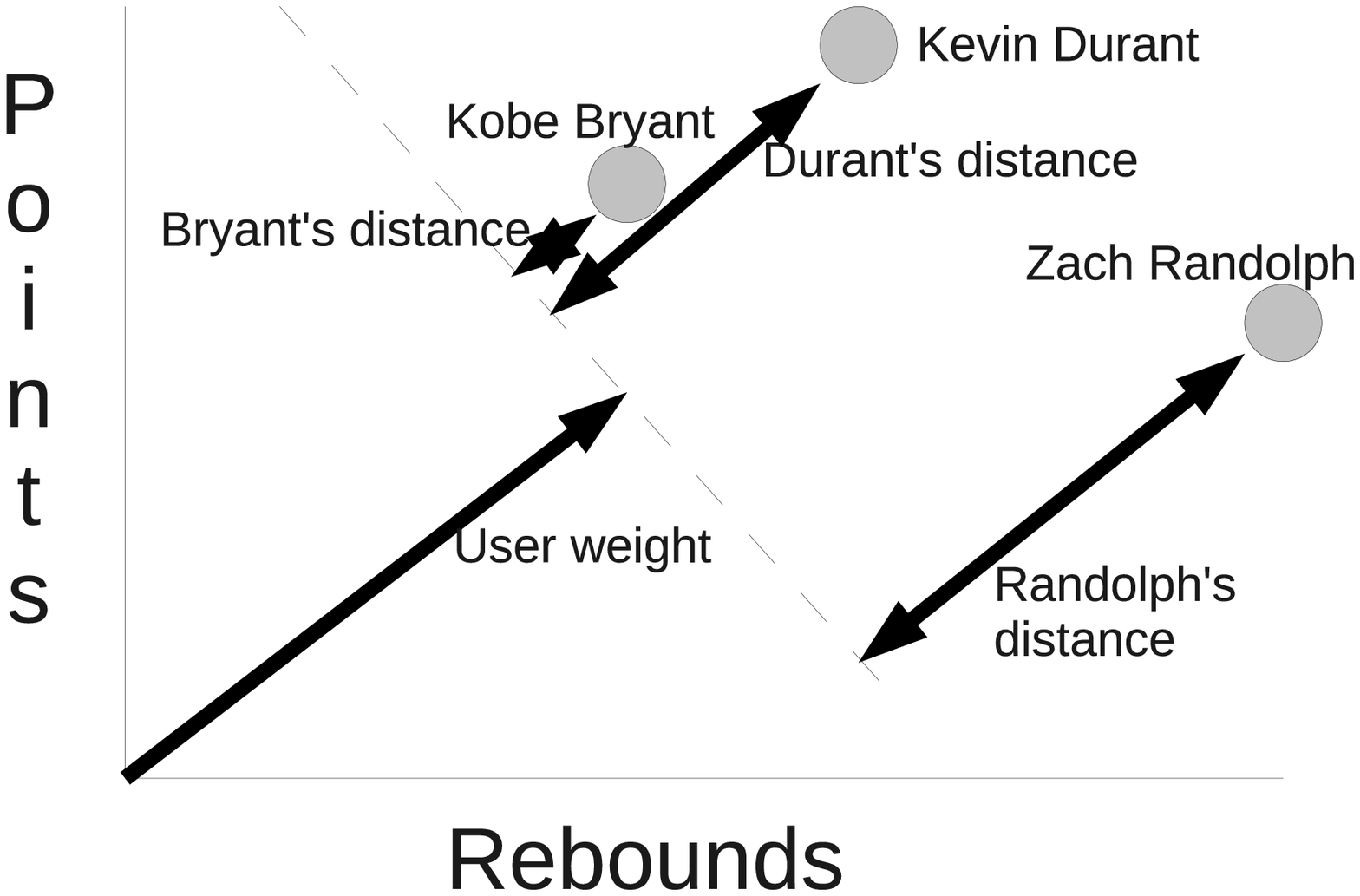}
 \caption{\label{fig:ranking_nba} A utility function $\uf(\mathrm{pts},\mathrm{rebs})=(\mathrm{pts}+\mathrm{rebs})/2$ represented as a vector $\vec{\uf}=\left<.5,.5\right>$, and three data points from Table~\ref{fig:nba} shown with their scores being proportional to the distance from the line orthogonal to $\vec{\uf}$.}
\end{figure}

Motivated to derive a succinct representation of a dataset, one with fixed cardinality, \nanon\ introduce {\em regret minimizing sets}~\cite{regretMin}, posing the question, 
``Does there exist one set of $\regretsize$ tuples that makes every user at least x\% happy?''  A regret minimizing set is a subset of a dataset that minimizes the regret ratio.  

A linear top-$k$ query can be considered as a problem of projection~\cite{vecProj}, where each tuple is regarded as a vector, as is the utility function.  The score of a tuple is proportional to the size of its projection onto the utility function vector and its scores for every possible utility function trace a (hyper-)sphere emanating from the point.  Minimizing regret ratio is equivalent to finding a subset that minimizes the maximum distance between the ``best'' spheres in the subset and the ``best'' spheres in the entire dataset, as illustrated in Figure~\ref{fig:regret_nba}.  Of the eight basketball players of Table~\ref{fig:nba}, \textit{Zach Randolph} achieves this criteria, so he is the optimal regret minimizing set of order (i.e., size) $1$.

\begin{figure}[t]
 \centering
 \includegraphics[scale=.35, clip=true, trim=0 16.5cm 7.5cm 0]{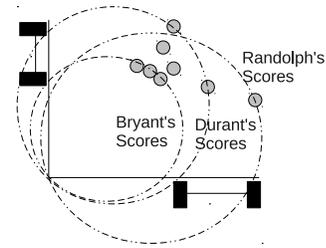}
 \caption{\label{fig:regret_nba}The spheres of scores for Durant, Bryant, and Randolph for {\em rebounds} as the $x$-attribute and {\em points} as the $y$-attribute.  The regret ratio of a subset $\ss$ is (roughly) the ratio of the distance from the best possible score in $\ss$ from the best possible score on the entire dataset.  For Durant and Bryant, this is maximized on the $x$-axis (where the best score is instead Randolph), and for Randolph, conversely, this is maximized on the $y$-axis (where the best score is instead Durant).}
\end{figure}

Randolph, the optimal regret minimizing tuple, however, is a peculiar choice to represent the dataset of Table~\ref{fig:nba} since he is the worst rated with respect to points.  This exposes a weakness of regret minimizing sets: they are based on assuming that a ``happy'' user is one who obtains their absolute top choice.  However, for an analyst curious to know what is a high point-scoring player, is he really dissatisfied with LeBron James as a query response rather than Kevin Durant?  

To change the scenario a bit, consider a dataset of hotels and a user searching for one that suits his preferences.  The absolute top theoretical choice may not suit him especially well at all.  It could be fully booked. Or, he might know that the manager reminds him of his ex-wife.  Regardless, it makes sense to present him a few, say $k$, options, any of with which he would be happy.  

We generalize the concept of regret and of regret minimizing sets to that of $k$-regret, analogous to the difference between top-$k$ queries and top-$1$ queries, because top-$k$ is often a better threshold for ``happiness''.  The analogous problem is to find a subset $\ss$ of points in the dataset that minimize the distance from the best point in $\ss$ to the $k$'th best point in the entire dataset.  This relaxation prevents having to fit an outlier tuple like Randolph.

%Another important consideration.  As Nanongkai\ea suggest, it is difficult to determine user weights.  It could very well be that the worst-case weight that max-regret is trying to minimize corresponds to a user who does not even exist!  Does anybody look at basketball statistics, and interest themselves solely in rebounds?  Furthermore, pessimistic view that we should appease the worst-case weight at the potential expense of all others.  So, we introduce an important complement, {\em total regret}, which is a summation rather than a maximum.  In that sense, we are trying to address the question: ``Does there exist one set of k tuples that, on average, makes every user happy?''

\subsubsection*{Optimality}
A fundamental open question remains with regards to both the problem introduced by \nanon\ and our generalisation of it.  How can one efficiently compute the optimal $k$-regret minimizing set of a predetermined cardinality $\regretsize$, the subset that achieves the minimal regret ratio of all size $\regretsize$ subsets of the dataset?  This is a problem conjectured to be \NPH\ by \nanon\ for $k=1$: it involves searching for the best among $\order(n^c)$ different subsets.

We introduce algorithms that strive to compute optimal $k$-regret minimizing sets.  
Towards this end, we relate the recent work on top-$k$ depth contours of Chester\ea~\cite{mrtop} for the reverse top-$k$ problem of Vlachou\ea~\cite{rtopVlachou}.  The top-$k$ depth contours are a dual space, geometric idea that succinctly represent exactly the $k$'th ranked tuple for all utility functions.  We demonstrate that these ideas are directly applicable to discovering optimal $k$-regret minimizing sets.  For instance, if the cardinality restriction is lifted, then the contours are {\em precisely} the optimal solutions (Lemma~\ref{thm:contour_optimal}).  In the presence of cardinality restrictions, the contours still aid in finding optimal solutions (Theorem~\ref{thm:alt_form}).  

%We further develop these ideas (Section~\ref{sec:technique}) to design efficient algorithms for finding optimal $k$-regret minimizing sets.  For the special case of two dimensions, we give a plane sweep algorithm that exactly computes the optimal solution efficiently (Section~\ref{sec:algo2d}).  For general dimension, we give a greedy, incremental algorithm that exploits our newfound ability to measure optimality in order to progress towards more optimal solutions (Section~\ref{sec:allD}).

\subsection{Contributions and Outline}\label{sec:contributions}
In this paper we propose the first algorithms for computing optimal $k$-regret minimizing sets.  In particular, we:
\begin{itemize}
 \item{generalize {\em regret} and {\em regret minimizing sets}, top-$1$ concepts, to those of {\em $k$-regret} and {\em $k$-regret minimizing sets}, top-$k$ concepts (Section~\ref{sec:prelims});}
 \item{identify that apparently unrelated work on top-$k$ depth contours and reverse top-$k$ queries~\cite{mrtop} sheds insight into the problem of identifying {\em optimal} $k$-regret minimizing sets (Section~\ref{sec:technique});}
 \item{introduce an $\order(n^2\regretsize)$ algorithm to compute the optimal size-$\regretsize$ $k$-regret minimizing subset $\ss$ of a two-dimensional dataset $\ds$ for the family of positive linear functions $\faml$, despite a conjecture by \nanon~\cite{regretMin} that the general dimension problem is intractible for $k=1$ (Section~\ref{sec:algo2d});}
 \item{introduce a greedy algorithm for general dimensions that leverages the relationship between top-$k$ depth contours and optimality in order to progress towards more optimal solutions (Section~\ref{sec:allD}); and}
 \item{relate our work within the context of other literature (Section~\ref{sec:lit}).}
\end{itemize}

\section{Preliminaries}\label{sec:prelims}
In the following sections, we will demonstrate how to compute optimal $k$-regret minimizing sets by equating the problem to one in dual space.  Within the dual space, the optimal solution is the top-$k$ depth contour if it small enough, or else the convex chain through the arrangement of lines in dual space that minimizes a particular distance ratio.  Before embarking on these objectives, however, we introduce some concepts formally in four subsections: first, $k$-regret and $k$-regret minimizing sets (Section~\ref{sec:regret_defs}); next, the transformation of the data into an arrangement of lines in dual space and some of the tools therein that we use (Section~\ref{sec:arrangement_defs}); penultimately, the top-$k$ depth contours that exist in the arrangement of lines and are fundamentally connected to finding optimal $k$-regret minimizing sets (Section~\ref{sec:contour_defs}); and, finally, the problem definition under study (Section~\ref{sec:prob_defs}).  Throughout the paper, we consider the family of positive linear functions $\faml$, which, without loss of generality, can be reduced to the family of positive unit linear functions $\famUnit$~\cite{vecProj}.  Nonetheless, we present the definitions for general family.

\subsection{k-Regret}\label{sec:regret_defs}
{\em $k$-Regret}, introduced here, is a generalisation of {\em regret}, introduced by \nanon~\cite{regretMin}.  We recall the definitions from that paper and introduce the generalisation in this subsection.

Given a dataset $\ds$ of $n$ $d$-dimensional numeric tuples, a subset $\ss\subseteq\ds$, a family of utility functions $\famf$, and a utility function $\uf\in\famf$:

\begin{definition}[gain~\cite{regretMin}]\label{def:gain}
The {\em gain} for a subset $\ss\subseteq\ds$ on $\uf\in\famf$ is:
$$\gain(\ss,\uf) = \max_{\point\in\ss}\uf(\point).$$
\end{definition}

That is to say, the {\em gain} of a subset $\ss$, given a utility function $\uf$, is simply the highest score achievable in $\ss$ for the function $\uf$.  Recalling the example of Table~\ref{fig:nba} and the utility function $f(\mathrm{pts},\mathrm{rebs})=(\mathrm{pts}+\mathrm{rebs})/2$, and assuming the data is normalized, the gain of $\{\mathrm{Bryant},\mathrm{Durant},\mathrm{Wade}\}$ is $0.828$.  The generalisation of {\em gain} is to {\em $k$-gain}:

\begin{definition}[$k$-gain]\label{def:jgain}
Consider a descending order list of $\uf(p)$ for all $p\in\ss\subseteq\ds$, given $\uf\in\famf$.  
Then, the {\em $k$-gain} of $\ss$ on $\uf$ is simply the $k$'th value in the list.
\end{definition}

In other words, the {\em $k$-gain} for a subset $\ss\subseteq\ds$ is the $k$'th best score achieved by a point in $\ss$ on the utility function $\uf$.  For the subset $\ss=\{\mathrm{Bryant},\mathrm{Durant},\mathrm{Wade}\}$ and the same function $\uf$, the $2$-gain is the second best score, $0.748$, the score for {\em Durant}.  For $k=1$, this definition is equivalent to Definition~\ref{def:gain}.  

{\em Regret ratio}, then, is a reflection of how well the gain of a subset approaches that of the entire dataset.

\begin{definition}[regret and regret ratio~\cite{regretMin}]\label{def:regret}
The {\em regret} for a subset $\ss\subseteq\ds$ on $\uf\in\famf$ is: $$\regret(\ss,\uf)=\gain(\ds,\uf)-\gain(\ss,\uf).$$
The {\em regret ratio} is: $$\rratio(\ss,\uf)=\frac{\regret(\ss,\uf)}{\gain(\ds,\uf)}.$$
\end{definition}

Since the best score for $\uf$ is $0.840$ the regret for the running example $\ss$ is $(0.840-0.828)$ and the regret ratio is $(0.840-0.828)/0.840$.  We generalise this to $k$-regret by evaluating how well the gain of a subset approaches the {\em $k$-gain} of the entire dataset.  Note, again, that this reduces to Definition~\ref{def:regret} if $k=1$.

\begin{definition}[$k$-regret and $k$-regret ratio]\label{def:kregret}
The {\em $k$-regret} is: $$\kregret(\ss,\uf)=\mathrm{max}(\kgain(\ds,\uf)-\gain(\ss,\uf), 0).$$
The {\em $k$-regret ratio} is: $$\krratio(\ss,\uf)=\frac{\kregret(\ss,\uf)}{\kgain(\ds,\uf)}.$$
\end{definition}

Since {\em Durant} is the second highest scoring tuple in $\ds$ for $\uf$, the $2$-regret ratio of $\ss=\{\mathrm{Bryant},\mathrm{Durant},\mathrm{Wade}\}$ is $(0.828-0.828)/0.828=0$.  The subset $\ss$ perfectly matches the top-$2$ requirement for utility function $\uf$.  Finally,

\begin{definition}[maximum $k$-regret ratio]\label{def:maxRegretRatio}
The {\em maximum $k$-regret ratio} for a subset $\ss\subseteq\ds$ with respect to a {\em family} of utility functions $\famf$ is: $$\krratio(\ss,\famf)=\supremum_{\uf\in\famf}\krratio(\ss,\uf).$$  
\end{definition}

The {\em maximum $k$-regret ratio} is the largest observable $k$-regret ratio for any utility function in an entire family.  For $\ss$ the $1$-regret ratio is maximized for $g(\mathrm{pts},\mathrm{rebs})=\mathrm{rebs}$, at which the best score obtainable is $\ss$ is $0.655$ and the $1$-regret ratio is $(1.000-0.655)/1.000$ and the $2$-regret ratio is $(0.771-0.655)/0.771$. 

Finally, a $k$-regret minimizing set of order $\regretsize$ is simply one with cardinality $\regretsize$ that minimizes the maximum $k$-regret ratio.  There exist {\em optimal} $k$-regret minimizing sets of order $\regretsize$, which are those that achieve minimal maximum $k$-regret ratio of all subsets of size $\regretsize$.

\begin{definition}[optimal $k$-regret minimizing set]\label{def:regretMin}
An {\em optimal $k$-regret minimizing set} of order $\regretsize$ on a dataset $\ds$ given a family of utility functions $\famf$ is: $$\mins_{\regretsize}(\ds,\famf)=\mathrm{argmin}_{\ss\subseteq\ds, |\ss|\leq \regretsize}\krratio(\ss,\famf).$$ 
\end{definition}

As well, Definition~\ref{def:regretMin} reduces to that of \nanon~\cite{regretMin} if $k=1$.

\subsection{Arrangements of Lines}\label{sec:arrangement_defs}
The algorithms that we propose in this paper are geometric in nature and operate on arrangements of hyperplanes in dual space.  Arrangements of hyperplanes (or lines, in two dimensions), are well studied in Computational Geometry and are induced by the intersections of a set of hyperplanes.

\begin{definition}[arrangement]\label{def:arrangement}
An {\em arrangement} of a set of $d$-dimensional hyperplanes $\mathcal{H}$, denoted $\mathcal{A}_{\mathcal{H}}$, is a partitioning of $\mathbb{R}^d$ into cells, edges, and vertices.  Each {\em cell} is a connected component of $\mathbb{R}^d \setminus \mathcal{H}$.  Each vertex is an intersection point of some $d$ hyperplanes in $\mathcal{H}$.  An edge is a line segment between two vertices of $\mathcal{A}$.
\end{definition}

We arrive at an arrangement of hyperplanes by applying the duality transform introduced by Chester\ea~\cite{mrtop}, which fixes an arbitrary positive real $\tau$ and converts every point $p_i\in\ds$ to a hyperplane $h_i$ by considering $p$ as a vector $\vec{p}$ and constructing the hyperplane $h_i$ to be all vectors $\vec{x}$ that solve $\vec{p}\cdot\vec{x}=\tau$.

\begin{figure}[bht]
 \centering
 \includegraphics[scale=.35, clip=true, trim=0 16.5cm 7.5cm 0]{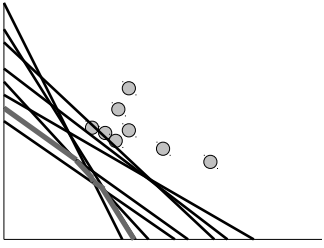}
 \caption{\label{fig:contour_nba}The eight basketball players from Table~\ref{fig:nba}, considering only the attributes \textit{points} and \textit{rebounds}, both normalized.  The tuples are transformed into translated nullspace equations, and the resulting arrangement of lines is shown.  Also depicted in thicker, light magenta lines is the second top-$k$ contour~\cite{mrtop}, a succinct representation of the $2$nd ranked tuples for any top-$k$ query.}
\end{figure}

\begin{definition}[translated nullspace transform~\cite{vecProj}]\label{def:duality}
Given a fixed positive real, $\tau$, and a dataset of $d$-dimensional points $\ds$, the {\em translated nullspace transform} transform each primal space point $p_i\in\ds$ into a dual space ($d$-$1$)-hyperplane $h_i$ (or line $l_i$ in two dimensions) composed of all vector solutions to the equation $\vec{p_i}\cdot\vec{x}=\tau$.
\end{definition}

For the basketball example, considering only the attributes \textit{points} and \textit{rebounds}, which have first been normalized, the arrangement of lines produced by the translated nullspace duality transform is illustrated in Figure~\ref{fig:contour_nba}.  Note that the intersection points of two lines $l_i$ and $l_j$ occur exactly in the direction of the vector $\vec{\uf}$ for which $\uf(p_i)=\uf(p_j)$.

Two other important concepts that are central ideas in Computational Geometry and of high relevance to this paper are {\em lower envelopes} of arrangements of lines and {\em convex chains} within arrangements of lines.

\begin{definition}[lower envelope]\label{def:envelope}
The lower envelope of an arrangement of lines is the set of edges under which no other edges exist.  
\end{definition}

For the purposes of this paper, in which we consider only the positive quadrant of Euclidean space, the {\em lower envelope} is the set of edges closest to the origin, $\origin$.  

\begin{definition}[convex chain]\label{def:chain}
A convex chain in an arrangement of lines $\mathcal{A}_{\ls}$ is the lower envelope in the arrangement of some subset $\ls'\subseteq\ls$ of lines, $\mathcal{A}_{\ls'}$.
\end{definition}

Alternatively, a convex chain can be considered to be any set of edges in the arrangement that form a convex polygon with $\origin$.

\subsection{Top-k Depth Contours}\label{sec:contour_defs}
We also recall here two definitions to establish what are top-$k$ depth contours, since they form the basis of our algorithms.

\begin{definition}[Top-$k$ Depth~\cite{mrtop}]\label{def:depth}
The {\em\dd} of a point $p$ within an arrangement $\mathcal{A}$, is the number of edges of $\mathcal{A}$ between $p$ and the origin.  That is to say, the depth of $p$ is the number of intersections between edges of $\mathcal{A}$ and $[\mathcal{O},p]$.\footnote{We remark that this is identical to the more familiar concept of a {\em level} if all the lines pass through the positive quadrant, as we assume here.  Nonetheless, we adopt this definition because there is no reason why the techniques described in this paper cannot be extended easily to handle attributes that range into negative values.}  Similarly, the {\em\dd} of a cell or edge of $\mathcal{A}$ is the \dd\ of every point within that cell.
\end{definition}

In their paper, Chester\ea~\cite{mrtop} show that the rank of a point in a dataset $\ds$ is precisely its \dd, and that \dd\ creates a series of $n$ contours in $\mathbb{R}^d$, the $i$'th of which is comprised of the transformed points that had rank exactly $i$ in $\ds$.

\begin{definition}[\dd\ contour~\cite{mrtop}]\label{def:contour}
A {\em\dd\ contour} is the set of edges in an arrangement $\mathcal{A}_{\mathcal{L}}$ that have \dd\ exactly $k$.
\end{definition}

\subsection{Problem Definition}\label{sec:prob_defs}
Now, we can formally describe the problem under study in this paper:

\medskip \noindent \textbf{Problem Definition 1.} Given any integer $\regretsize$ and set $\ds$ of $n$ $d$-dimensional points, find an optimal $k$-regret minimizing set of order $\regretsize$, $\mins_{\regretsize}(\ds,\famUnit)$, for the family of positive unit linear functions $\famUnit$.\medskip

\section{A contour view of regret}\label{sec:technique}
In this section, we show that the concept of regret that was introduced by Nanongkai\ea~\cite{regretMin}\---and also the generalisation we introduce in this paper\---are strongly connected to the dual space concept of top-$k$ depth contours introduced by Chester\ea~\cite{mrtop}.  More precisely, we prove Theorem~\ref{thm:alt_form} which equates the problem of finding an optimal $k$-regret minimizing set to a dual space problem of finding a set of lines that are ``closest'' to the top-$k$ depth contour.  This alternative formulation of the problem facilitates designing algorithms in the dual space for two dimensions (Section~\ref{sec:algo2d}) and general dimension (Section~\ref{sec:allD}) to find optimal regret minimizing sets.

The argument proceeds by showing that the contour itself, $\contour_k$, is the optimal solution, provided that it is small enough, $|\contour_k|\leq \regretsize$ (Lemma~\ref{thm:contour_optimal}).   In the dual space, the regret ratio of a line relative to another line, given a utility function $f$, is given by the relative distances of the lines from the origin in the direction indicated by $f$ (Lemma~\ref{thm:regret_is_distance}).  So, the evaluation of regret in dual space is a scaled Euclidean distance computation.

We also show that the best options available to users within a set of points $\ss\subseteq\ds$ is exactly given in the the set of dual lines of $\ss$, their lower envelope (Lemma~\ref{thm:envelope_best_of_ss}).  So, minimizing the scaled distance of that envelope from the contour yields an optimal solution (Theorem~\ref{thm:alt_form}).

\begin{lemma}\label{thm:contour_optimal}
The set of points contributing to $\contour_k$ is a $k$-regret minimizing set $\mins_{\regretsize}(\ds,\faml)$ if $|\contour_k|\leq\regretsize$.
\end{lemma}
\begin{proof}
$\contour_k$ is constructed such that, for any linear function $\uf\in\faml$, a point $p$ on $\contour_k$ has rank exactly $k$.  Therefore, $\krratio(\{p\},\uf)=0, \forall\uf\in\faml$.
\end{proof}

To summarize Lemma~\ref{thm:contour_optimal}, the contour is necessarily an optimal solution for any $\regretsize\geq|\contour_k|$ because it is a representation of the $k$'th ranked tuple for any linear utility function, and the $k$-regret ratio of the $k$'th ranked tuple is $0$.  

Since the $k$-contour represents the barrier of no $k$-regret, any points transformed to lines farther from the origin $\origin$ than the contour have positive regret proportional to the distance from $\origin$.  Conversely, any points transformed to lines closer to the origin than $\contour_k$ with respect to $\uf$ have $k$-regret=0, because they are within the top-$k$ on $\uf$.

\begin{lemma}\label{thm:regret_is_distance}
For any utility function $\uf\in\famUnit$ and tuple $p_i\in\ds$ transformed to line $l_i\in\ls$, let $\Delta_{\contour_k}$ denote the distance of $\contour_k$ from $\origin$ with respect to $\vec{\uf}$, $\Delta_i$, the distance of $l_i$ from $\origin$, and $\Delta'\geq 0$ denote the distance of $l_i$ from $\contour_k$.  Then, $\krratio(\{p_i\})=\Delta'/\Delta_i$.
\end{lemma}
\begin{proof}
Recall that each line $l_i\in\ls$ is constructed of vectors $\vec{x}$ such that $\vec{p_i}\cdot\vec{x}=\tau$.  So, since $||\vec{\uf}||=1$, the distance of $l_i$ from $\origin$ in the direction of $\vec{\uf}$ is $\tau/\uf(p_i)$.  

Thus, if $p_{\contour_k}\in\ds$ represents a point on $\contour_k$ in the direction of $\vec{\uf}$, 
\begin{eqnarray*}
\Delta'&=&\Delta_i-\Delta_{\contour_k}\\
&=&\frac{\tau}{\uf(p_i)}-\frac{\tau}{\uf(p_{\contour_k})}\\
&=&\tau\frac{\uf(p_{\contour_k})-\uf(p_i)}{\uf(p_i)\uf(p_{\contour_k})}\\
&=&\Delta_i(\krratio(\{p_i\}))
\end{eqnarray*}
\end{proof}

\begin{corollary}\label{thm:dist_with_fixed_uf}
For a fixed $\uf\in\famUnit$, let $\Delta_i$ denote the nonnegative distance for some line $l_i\in\ls$ from $\contour_k$ in the direction of $\vec{\uf}$.  Then, $\Delta_i\propto\krratio(\{p_i\},\uf)$.
\end{corollary}

Lemma~\ref{thm:regret_is_distance} establishes that the the regret for a singleton set $\{p_i\}$ on a function $\uf$ in dual space is related to the Euclidean distance of the transformed line $l_i$ to $\origin$ and to $\contour_k$.  Corollary~\ref{thm:dist_with_fixed_uf} notes that, if we consider only a single utility function, then the regret for different singleton sets can be straightforwardly compared by the distance from $\contour_k$, since the distance of $\contour_k$ to $\origin$ is static.  Next, we show that, given a non-singleton set, the maximum $k$-regret can be evaluated efficiently, because it is to be observed on the lower envelope in the dual space.

\begin{lemma}\label{thm:envelope_best_of_ss}
For a set $\ss\subseteq\ds$, let $\ls$ denote the set of lines produced by transforming every point $p_i\in\ss$ into its translated nullspace, $l_i$.  The lower envelope of $\ls$ captures the maximum gain (and, ergo, minimum regret) of $\ss$ for any $\uf\in\famUnit$.
\end{lemma}
\begin{proof}
For any $\uf\in\famUnit$, the nearest line to $\origin$ is that line $l$ which is on the lower envelope in the direction of $\vec{\uf}$.  Since $l$ has the smallest distance to $\origin$ of all lines in $\ls$, it also has the smallest distance (possibly negative) to $\contour_k$ of all lines in $\ls$.  By Corollary~\ref{thm:dist_with_fixed_uf}, $p_i$ then also has the minimum regret with respect to $\uf$ of all $p\in\ss$.
\end{proof}

Furthermore, while the lower envelope captures all the interestingness of a set $\ss$, there are only particular points on the envelope where the maximum regret can occur: the points where the distance ratio could be maximized.  Lemma~\ref{thm:maxDistAtVertex} confirms that these points are exactly the vertices of the lower envelope and the vertices of the contour.

\begin{lemma}\label{thm:maxDistAtVertex}
Consider a contour $\contour_k$ and a convex chain of line segments $\mathcal{C}$.  Let $\Delta_i$ denote the distance of $\cchain$ from $\origin$ in the direction of $\vec{\uf}$ and, similarly, $\Delta'$, the distance of $\cchain$ to $\contour_k$.  The expression $\Delta'/\Delta_{i}$ is maximized either at a vertex of $\contour_k$ or a vertex of $\mathcal{C}$.
\end{lemma}
%\begin{IEEEproof}
\begin{proof}
Both $\contour_k$ and $\mathcal{C}$ are piecewise linear, so the expression $\Delta'/\Delta_{i}$ can only be maximized at some junction point.
\end{proof}
%\end{IEEEproof}

Since Lemmata~\ref{thm:envelope_best_of_ss}~and~\ref{thm:maxDistAtVertex} permit evaluating regret in the dual space of translated nullspaces, we can derive an alternative view of the problem of finding an optimal $k$-regret minimizing set, as shown in Theorem~\ref{thm:alt_form} and illustrated in Figure~\ref{fig:extensions}.

\begin{figure}[bth]
 \centering
 \includegraphics[scale=.285, clip=true, trim=0 14cm 5cm 0]{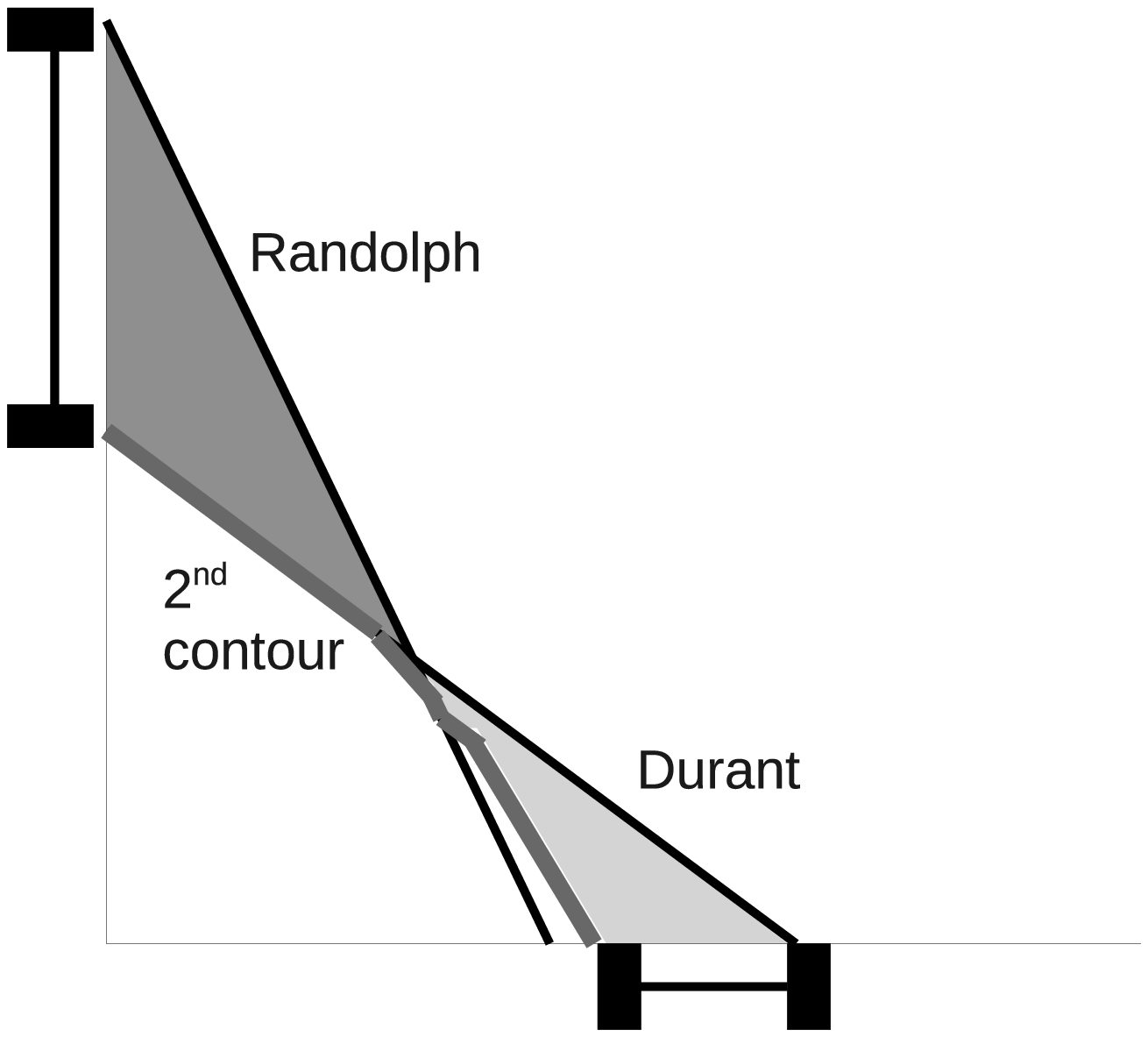}
 \includegraphics[scale=.285, clip=true, trim=0 14cm 5cm 0]{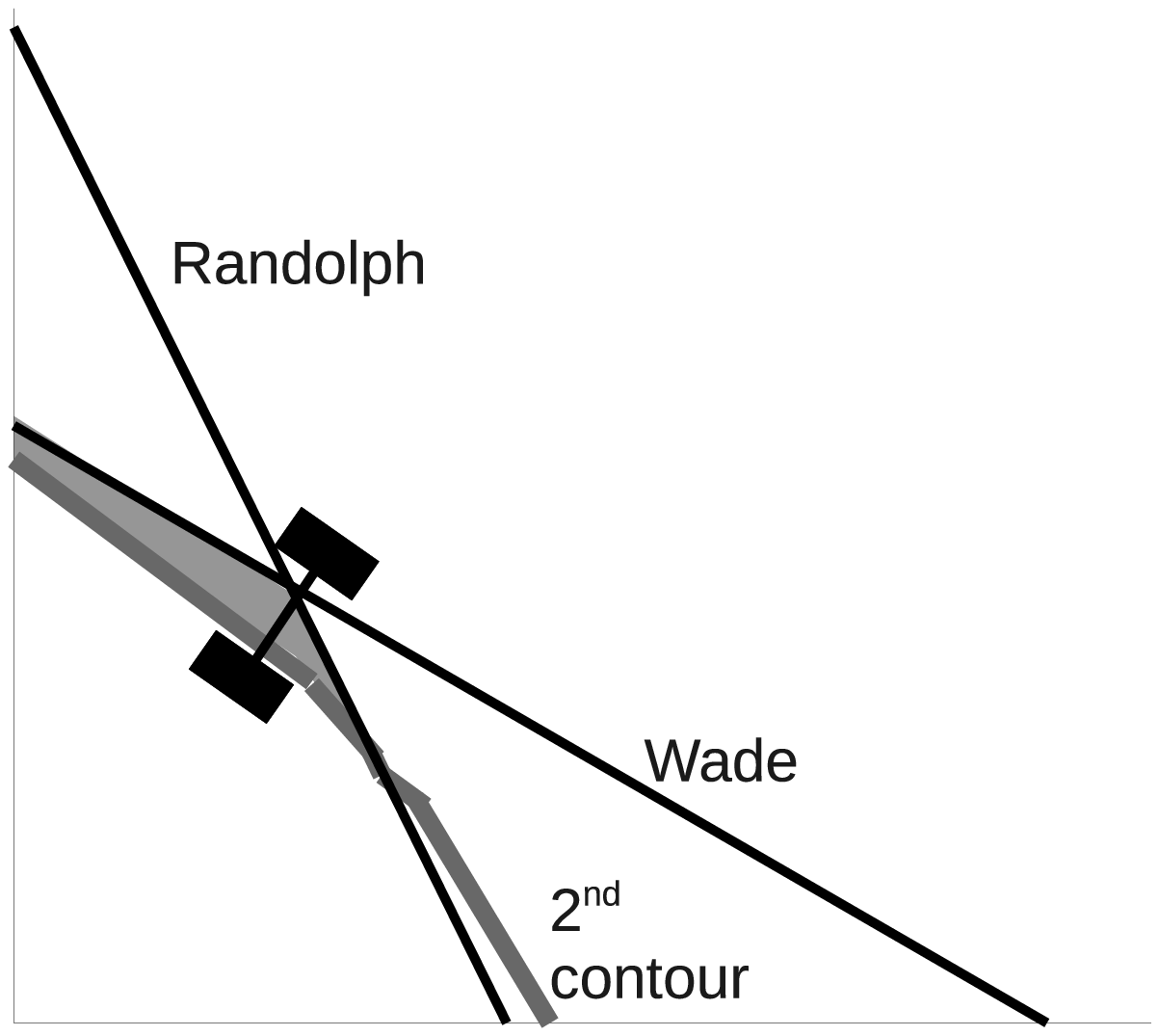}
 \caption{\label{fig:extensions} An illustration of $2$-regret.  Shown left are two different order-$1$ sets, $\{\mathrm{Durant}\}$ and $\{\mathrm{Randolph}\}$, along with the second contour.  The axes are the directions of maximum regret for the respective sets.  To the right is shown the (non-optimal) order-$2$ set $\{\mathrm{Randolph}, \mathrm{Wade}\}$. }
\end{figure}

\begin{theorem}\label{thm:alt_form}
Let $g$ denote the function which transforms any point $p\in\ds$ to its translated nullspace and let $\contour_k$ denote the top-$k$ depth contour of $\ds$.  The optimal $k$-regret minimizing set $\ss$ of $\ds$ with size at most $\regretsize$ on the family of positive unit linear functions $\famUnit$ is exactly the set of lines $\ls=\{g(p), p\in\ss\}, |\ls|\leq\regretsize$, the lower envelope $\mathcal{E}$ of which minimizes the maximum ratio of distances from $\env$ to $\contour_k$ and $\env$ to $\origin$ at any vertex of $\env$ and of $\contour_k$.
\end{theorem}
\begin{proof}
From Lemma~\ref{thm:regret_is_distance}, the regret ratio for each point in $\ss$ is equivalent to the ratio of $\env$ to $\contour_k$ and $\env$ to $\origin$, and from Lemma~\ref{thm:envelope_best_of_ss}, the best such ratio is on the lower envelope of $\ls$.  From Lemma~\ref{thm:maxDistAtVertex}, this must occur at either a vertex of $\contour_k$ or a vertex of $\env$.
\end{proof}

A final remark is with regard to the two dimensional case, in particular.  We note that any lower envelope is, in fact, a convex chain, so the two dimensional problem can be viewed rather as searching for the best convex chain.

\begin{lemma}\label{thm:chains_as_solns}
Let $g$ denote the function which transforms any point $p\in\ds$ to its translated nullspace and let $\contour_k$ denote the top-$k$ depth contour of $\ds$.  The optimal $k$-regret minimizing set $\ss\subseteq\ds\subseteq\mathbb{R}^2$ with size at most $\regretsize$ is exactly the convex chain $\mathcal{C}$ through the arrangment of lines $\ls=\{g(p), p\in\ss\}$ that has at most $\regretsize-1$ turns and that minimizes the maximum ratio of the distance from $\mathcal{C}$ to $\contour_k$ and the distance of $\mathcal{C}$ to $\origin$.
\end{lemma}
\begin{proof}
Note that any lower envelope of a set of lines is, in fact, a convex chain and that, by convexity, any line can appear at most once on the lower envelope.  So, a convex chain with $\regretsize$ lines will have $\regretsize-1$ turns.  Also, any convex chain with $\regretsize-1$ turns that is optimal must be a lower envelope of some set of lines, otherwise the sequence of turns that follows the lower envelope of the same set of lines will be more optimal.  By Theorem~\ref{thm:alt_form}, this is the optimal $k$-regret minimizing set of size $\regretsize$ for $\ds$.
\end{proof}

\section{An algorithm for two\\ dimensions}\label{sec:algo2d}

\begin{figure*}[tbh]
 \centering
 \subfigure[The arrangement of lines after initialisation]{
  \includegraphics[scale=.35, clip=true, trim=0 14.5cm 5cm 0]{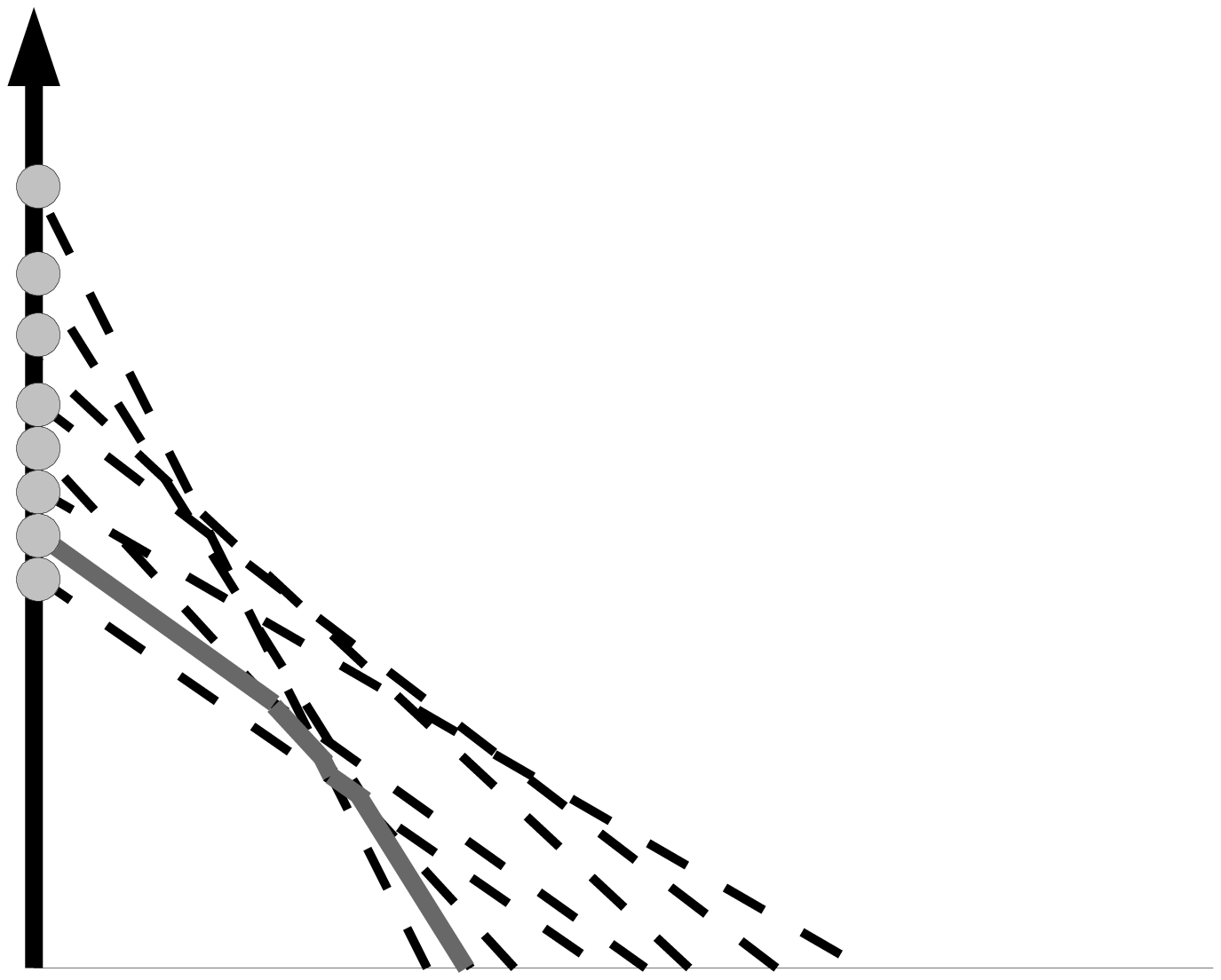}
 }
\hspace{1cm}
 \raisebox{2.5cm}{\subtable[The priority queue, including all initial intersections]{
 \begin{tabular}{l || r}
\multicolumn{2}{c}{$\pq$}\\\hline
\textit{Lines}&\textit{slope of $r$}\\\hline
(Wade, Nowitzki)&35.71\\
(Anthony, Stoudemire)&15.39\\
(Bryant, Anthony)&6.07\\
(Stoudemire, Randolph)&2.64\\\hline
 \end{tabular}
 }}
\hspace{1cm}
 \raisebox{2.65cm}{\subtable[The initial cost values for each line, the difference in y-intercept to the contour]{
 \begin{tabular}{l || l | l}
\multicolumn{3}{c}{$\as$}\\\hline
\textit{Line}&$m=1$&$m=2$\\\hline
Durant&0.000&0.000\\
James&0.000&0.000\\
Wade&0.057&0.057\\
Nowitzki&0.062&0.062\\
Bryant&0.080&0.080\\
Anthony&0.089&0.089\\
Stoudemire&0.105&0.105\\
Randolph&0.188&0.188\\\hline
 \end{tabular}
 }}
 \caption{\label{fig:alg2d_init} The initialisation of the algorithm for two dimensions.  The sweep ray begins at the $y$-axis.  The priority queue $\pq$ is populated with intersection points of lines that neighbour on the $y$-axis and intersect in the positive quadrant, sorted by the order in which $r$ will pass through them.  Every entry in the path list $\as$ is originally set to the distance of the corresponding line to the contour along the $y$-axis. }
\end{figure*}

As we showed in Lemma~\ref{thm:chains_as_solns}, solving the regret minimizing problem reduces to finding the best convex chain $\mathcal{C}$ with fewer than $\regretsize$ turns through an arrangement.  The optimal solution is the one which minimizes the distance ratio of $\cchain$ to $\contour_k$ and $\cchain$ to $\origin$.  There are potentially $n\choose \regretsize$ different convex chains with at most $\regretsize-1$ turns; however, so we need to improve upon the $\order(n^\regretsize)$ naive algorithm which tries every combination.

We offer here a plane sweep, dynamic programming algorithm which runs $\order(\regretsize n^2)$ and independently of $k$.  The algorithm follows each of the $n$ translated nullspace lines in $\ls$ radially from the $y$-axis to the $x$-axis along a sweep line $r$, processing intersection points, and remembering $\regretsize$ of the best paths yet encountered for each line.  We maintain three data structures, each which maintains some invariant with respect to the current position of $r$.  The optimal solution can then be read from one of the data structures at the conclusion of the radial plane sweep.

As one scans the plane, the data structures need to be updated whereever two lines intersect in order to reflect the changed state of the arrangement of lines and to maintain the invariants.  All three of our data structures need to be updated at and only at line intersections.

\subsection{Data Structures}\label{sec:structures}
As mentioned, there are three data structures.  The first is the set of lines, $\ls$, sorted by their distance from the origin in the direction of $r$.  The second is a priority queue, $\queue$, containing intersection (or {\em event}) points yet to be processed, sorted by the order in which $r$ will pass through them.  The third data structure, $\as$ is for the dynamic programming component and maintains for each line the best solutions seen between the $y$-axis and the current position of $r$.  The structure $\as$ is an $n$ by $\regretsize$ matrix.   In each cell $(i,j)$ is stored the optimum convex chain from $r$ back to the $y$-axis which both contains at most $j$ lines and ends on line $l_i$.

\subsubsection{Data structure transitions}\label{sec:transitions}
\begin{figure}[b!]
 \centering
 \includegraphics[scale=.285, clip=true, trim=0 14.5cm 5cm 0]{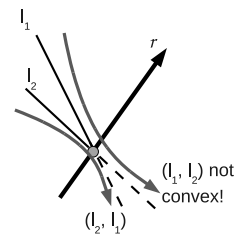}
 \caption{\label{fig:alg2d_inset} An illustration of the three possible paths through an intersection point.  Either line $l_1$ or $l_2$ could simply pass through.  Because an envelope of lines must form a convex chain, on the other hand, only $l_2$ has the luxury of turning onto $l_1$.  The path ($l_2, l_1$) requires an ``illegal'' concave turn. }
\end{figure}

We first describe the algorithm by how the data structures evolve at each intersection point.  Both $\queue$ and $\ls$ behave as in traditional plane sweep algorithms, whereas $\as$ behaves as in a traditional dynamic programming algorithm.

\bigskip
\subsubsection*{$\ls$}
Consider an intersection point $p_{i,j}$, the intersection of lines $l_i$ and $l_j$.  Because the lines are intersecting, we know they are immediately adjacent in $\ls$.  We swap $l_i$ and $l_j$ in $\ls$ to reflect the fact that immediately after $p_{i,j}$, they will have opposite order as compared to beforehand.

\subsubsection*{$\queue$}
The priority queue contains all those intersection points that are between $r$ and the positive $x$-axis that feature two lines which have been adjacent at some point between the positive $y$-axis and $r$.  Again, consider an intersection point $p_{i,j}$.  Immediately thereafter, lines $l_i$ and $l_j$ have been swapped.  As such, there are potentially two new intersection points to add to $\queue$, namely $l_i$ and his new neighbour (should one exist) and $l_j$ and his new neighbour (again, should one exist).  Both these intersection points are added to the appropriate place in $\queue$, provided that they are between $r$ and the positive $x$-axis.  The point $p_{i,j}$ is removed.

\begin{figure*}[tbh]
 \centering
 \subfigure[The arrangement of lines immediately before processing the darker point, (Stoudemire, Randolph).]{
  \includegraphics[scale=.35, clip=true, trim=0 14.5cm 5cm 0]{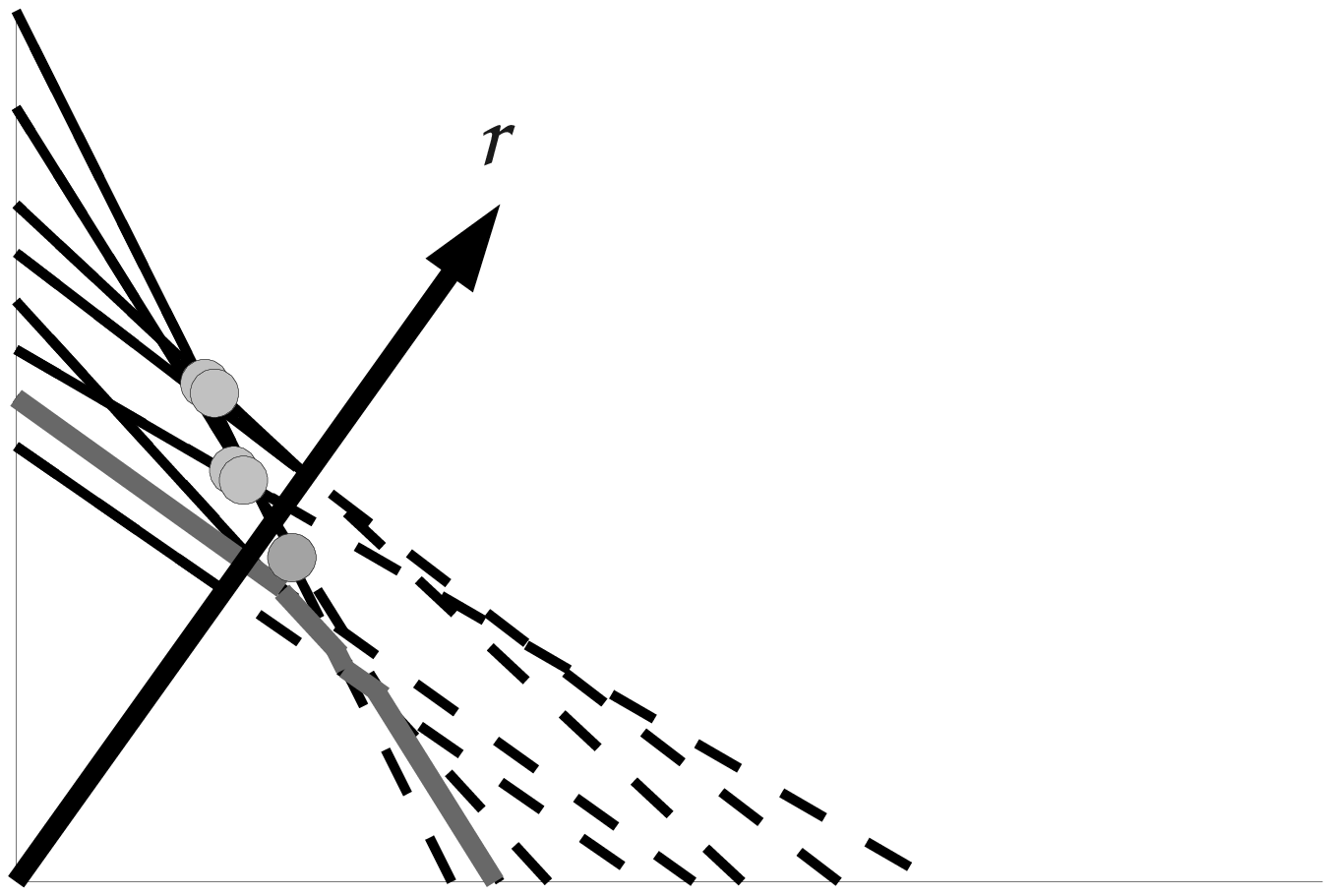}
 }
\hspace{.5cm}
 \raisebox{2.5cm}{\subtable[The priority queue, including all intersection points between $r$ and the $x$-axis which have been discovered between the $y$-axis and $r$.  The grayed entry is the one being processed.  The bolded entry is the one that is newly added by processing this point.]{
 \begin{tabular}{l || r}
\multicolumn{2}{c}{$\pq$}\\\hline
\textit{Lines}&\textit{slope of $r$}\\\hline
{\color{grayish}(Stoudemire, Randolph)}&{\color{grayish}2.64}\\
\textbf{(Nowitzki, Randolph)}&\textbf{2.48}\\
(Nowitzki, Stoudemire)&2.22\\
(Wade, Anthony)&2.07\\
(James, Nowitzki)&0.74\\
(Wade, Bryant)&0.62\\
\hline
 \end{tabular}
 }}
\hspace{.5cm}
 \raisebox{2cm}{\subtable[The best cost values for each line as of the processing of the darker point in (a).  The bold indicates the sustained value in the table after processing the point.]{
 \begin{tabular}{l || l | l}
\multicolumn{3}{c}{$\as$}\\\hline
\textit{Line}&$m=1$&$m=2$\\\hline
Durant&0.000&0.000\\
James&0.00&0.000\\
Wade&0.057&0.057\\
Nowitzki&0.062&0.059 (Wade)\\
Bryant&0.080&0.080\\
Anthony&0.089&0.084 (Bryant)\\
Stoudemire&0.105&0.066 (Wade)\\
Randolph&0.188&\textbf{0.083 (Wade)}\\\hline
 \end{tabular}
 }}
 \caption{\label{fig:alg2d_prog} The processing of the tenth event point, the intersection of the lines corresponding to Stoudemire and Randolph.  The distance along $r$ of the intersection point from the contour is $0.034$.  In this case, the newly discovered length-$2$ chain, (Stoudemire, Randolph), has cost $max(0.105, 0.034)$, which does not improve on the value already found for Randolph, $0.083$.}
\end{figure*}

\subsubsection*{$\as$}
Again, consider an intersection point $p_{i,j}$ featuring lines $l_i$ and $l_j$.  Let $l_i$ be farther from the origin than $l_j$ in the direction of a ray {\em after} $r$.  There are three valid paths through $p_{i,j}$, as illustrated in Figure~\ref{fig:alg2d_inset}.  

First consider the line, $l_i$, that emerges above after $p_{i,j}$ (line $l_2$ in Figure~\ref{fig:alg2d_inset}).  Consider also the row of $\regretsize$ cells of $\as$ that describe best paths for $l_i$.  The $h$'th such cell, describes the chain with optimum distance to $\contour_k$ that uses at most $h-1$ turns and emerges from $p_{i,j}$ along line $l_i$.  Because the turn $(l_j,l_i)$ is invalid, paths for $l_i$ cannot change, only their costs.  The cost is updated to the larger of what the value was before and the distance from $p_{i,j}$ to $\contour_k$ in the direction of $r$.
  
So, considering first a chain that leaves along $l_j$, the best possible route to the next intersection point of $l_j$ is exactly whatever was the best possible route along $l_j$ to $p_{i,j}$.  The cost of that route is the larger of the distance from $p_{i,j}$ to $\contour_k$ and the cost of the best possible route to the previous intersection point of $l_j$.  This value is updated for each of the $m$ cells in row $i$.

On the other hand, for a chain that emerges along $l_j$, there are two possibilities, depending on the best route to get to $p_{i,j}$.  Specifically, the best route of $h$ turns to $p_{i,j}$ is the cheaper of the best route to the previous intersection point along $l_j$ that used $h$ turns and the best route to the previous intersection point along $l_i$ that used $h-1$ turns.  The final cost is then the larger of the distance from $p_{i,j}$ to $\contour_k$ and the minimum cost to $p_{i,j}$ as just described.

\subsection{Algorithm Description}

As we have hinted, the algorithm is a plane sweep through the arrangement of lines, searching for the minimal cost convex chain with fewer than $\regretsize$ turns.  The sweep features a ray $r$, originally positioned on the positive $y$-axis, moving radially through the positive quadrant to the positive $x$-axis.  This plane sweep approach is appropriate as a consequence of Lemma~\ref{thm:maxDistAtVertex}, which reveals that the cost of any convex chain is maximized at an event point.

To initialize the algorithm, the intersection point of every line with the positive $y$-axis is processed, populating the data structures as in Figure~\ref{fig:alg2d_init}.  We only add to the priority queue the intersection points of lines that are immediate neighbours with respect to a sort on $y$-intercept, and for which the intersection point occurs in the positive quadrant.  These points are maintained in the queue in descending order of angle from the positive $x$-axis (i.e., in the order in which the ray $r$ will ``sweep'' them).  The array $\as$ is initialised with empty paths for every cell, with a cost set to the distance of the relevant line to the contour.

The algorithm proceeds simply by popping the next event from $\pq$, updating the data structures as per Section~\ref{sec:transitions}, and pushing the new event points onto the queue.  For the running basketball example, this is illustrated in Figure~\ref{fig:alg2d_prog}.  For event points that correspond to vertices of the contour (since these, too, are intersection points of lines that will eventually be discovered by the plane sweep), we update every cell of $\as$ with new maximum costs for each line that has become more distant from $\contour_k$.

\begin{figure*}[tbh]
 \centering
 \subfigure[The arrangement of lines after processing every event point, emptying $\queue$.]{
  \includegraphics[scale=.35, clip=true, trim=0 14.5cm 5cm 0]{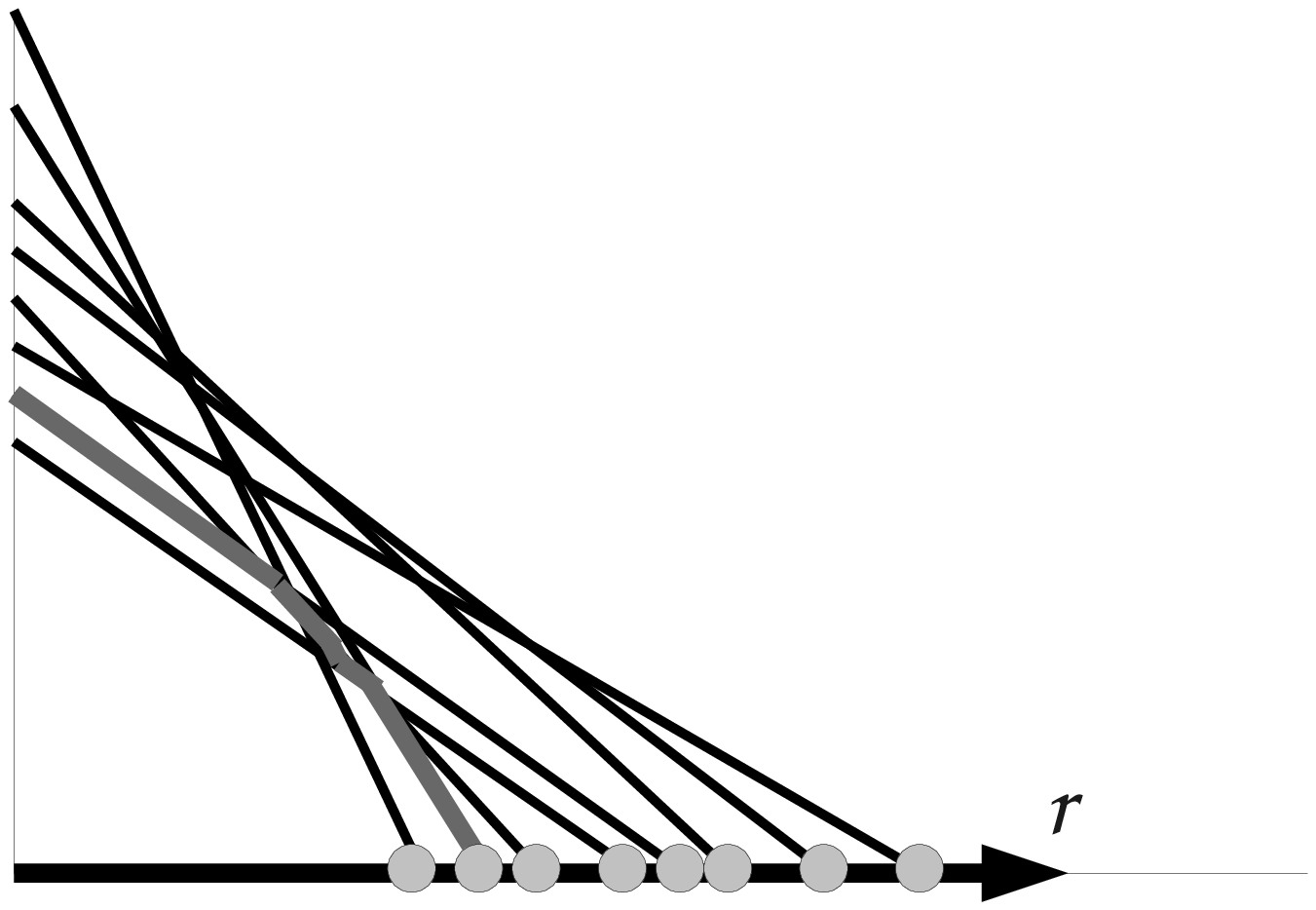}
 }
\hspace{1cm}
 \raisebox{2cm}{\subtable[The best cost values for each line at the conclusion of the plane sweep.  The bold entries represent optimal solutions for each $j\leq m$.]{
 \begin{tabular}{l || l | l}
\multicolumn{3}{c}{$\as$}\\\hline
\textit{Line}&$m=1$&$m=2$\\\hline
Durant&0.114&0.114\\
James&0.208&0.208\\
Wade&0.625&0.625\\
Nowitzki&0.117&0.059 (Wade)\\
Bryant&0.566&0.566\\
Anthony&0.397&0.081 (Wade)\\
Stoudemire&\textbf{0.105}&\textbf{0.000 (Durant)}\\
Randolph&0.188&\textbf{0.000 (James)}\\\hline
 \end{tabular}
 }}
 \caption{\label{fig:alg2d_concl} The data structures at the conclusion of the plane sweep, when both $r$ reaches the $x$-axis and $\pq$ is empty.  Each cell $(h, i)$ of $\as$ contains the optimum solution that contains $i$ or fewer lines and ends with line $l_h$.  The minimum value in the entire table is the optimal solution for this two-dimensional $j$-Regret problem. }
\end{figure*}

Eventually, $\pq$ will be exhausted as $r$ reaches the positive $x$-axis (see Figure~\ref{fig:alg2d_concl}).  Every cell is updated with new maximum costs at this last contour vertex (the intersection of the contour with the $x$-axis).  The final step is to scan through all of $\as$ and determine the smallest cost.  This is the optimal solution, which is reported along with the path used to obtain it.

Algorithm~\ref{alg:2d} describes the algorithm with greater detail.

\subsection{Asymptotic Complexity}\label{sec:2d_complexity}

\begin{theorem}
Algorithm~\ref{alg:2d} for the two-dimensional case finds an optimal $k$-regret minimizing set of order $\regretsize$ in $\order({\regretsize}n^2)$ time with $\order(n^2)$ space.% if solving for maximum regret and in $\order(mn^2\log n)$ time with $\order(n^2)$ space if solving for total regret.
\end{theorem}

\begin{proof}
First consider space used.  
The size of the contour is bounded by $n$.  Of the three data structures, $\ls$ is of size exactly $n$, $\as$ is of size exactly $n*\regretsize\leq n^2$, and $\pq$ is proportional to the largest number of intersection points that have been discovered but not processed, clearly less than $n*(n-1)$.  Therefore, the total space is $\order(n^2)$.

Regarding time, %consider first the maximum regret problem.  F
for each non-contour event point, of which there may be up to $n^2$, up to $2\regretsize$ cells of $\as$ are updated.  For each contour event point, of which there are $|\contour_k|\leq n$, each of the $n\regretsize$ cells of $\as$ could potentially be updated.  The initialisation requires a sort of $n$ lines and then an initialisation of up to $n-1$ insertions into $\pq$ and $n\regretsize$ values for $\as$, which can be computed in constant time.  At the conclusion of the plane sweep, all $n\regretsize$ cells of $\as$ must be scanned.  Therefore, the entire procedure is $\order(n^2\regretsize)$.

%For the total regret problem, the difference is in the computation of the {\em cost} at each of the $2m$ cells at each of the $n^2$ event points.  Whereas for maximum regret, this is a constant time operation, computing area requires the $\order(\log |\contour_k|)=\order(\log n)$ lookup described in Section~\ref{sec:computeScore}.
\end{proof}

%\subsection{Lines not in general position}\label{sec:multi_intersections}
%Parallel lines do not affect our algorithm, but we do need to describe how to account for multiple lines intersecting simultaneously.

\section{An algorithm for general dimension}\label{sec:allD}
In Section~\ref{sec:algo2d}, we gave an efficient plane sweep, dynamic programming algorithm to find optimal $k$-regret minimizing sets in two dimensions.  Unfortunately, plane sweep algorithms do not readily generalize to higher dimensions.  So, in this section, we offer a greedy algorithm which exploits Lemma~\ref{thm:greedy} in order to progress towards an optimal solution.

\medskip
\begin{lemma}\label{thm:greedy}
Consider the envelope produced by some set $\ls$.  For another set $\ls'$ to better minimize regret, it is necessary that some line of $\ls'\setminus\ls$ either passes through the area between the contour and the envelope or remains entirely under the contour at the angle for which the distance ratio for $\ss$ is maximized.
\end{lemma}  
\begin{proof}
This results from Corollary~\ref{thm:dist_with_fixed_uf}, which indicates that for that fixed utility function, distance is directly proportional to the regret ratio.  The $L_2$ distance metric is transitive.  So, if no line in $\ls'$ intersects the area underneath the lower envelope of $\ls$ at the angle for which the distance ratio is maximized, then its maximum distance is clearly larger.
\end{proof}
\medskip

We use the insight of Lemma~\ref{thm:greedy} to design a greedy algorithm which behaves as follows.  Consider a given set of lines $\ss$ for which the distance ratio is maximized at point $p$.  Consider also a new line $l\not\in\ss$.  We advance to a new set $\ss'$ exactly when $l$ passes between the origin and $p$ and then if there is some element $l'\in\ss$ such that the maximum distance ratio of $\ss\setminus\{l'\}\cup\{l\}$ to $\contour_k$ is less than that at $p$.  That is to say, if $l$ intersects the segment $[\origin,p]$, then we look for a new, better set that can be obtained by replacing some element of $\ss$ with $l$.  We know from Lemma~\ref{thm:greedy} that the intersection test is a necessary condition to find some improved solution $\ss'$.

Overall, the algorithm begins with an initial seed solution and cycles through all lines repeatedly, conducting the above test in order to refine $\ss$, until no line can improve the cost any more.  We report this solution.  The greedy algorithm is feasible because, given the contour insight that we have developed in this paper, one can efficiently determine which of two sets is more optimal: it is the one for which the maximum distance of the lower envelope of the set from $\contour_k$ is minimized.  The algorithm is detailed in Algorithm~\ref{alg:allD}.

Note that the algorithm is guaranteed to terminate because it will always progress towards a better solution due to Line~13, and there are finitely many subsets of $\ls$.

\newpage
\section{Related Work}\label{sec:lit}
The idea of representing an entire dataset with a few representative tuples for multi-criteria decision making has drawn much research attention in the past decade, since the introduce of the Skyline operator by B\"{o}rzs\"{o}nyi\ea~\cite{skyline}.  However, its susceptibility to the curse of dimensionality is well-known.  Chan\ea~\cite{highsky} made a compelling case for this, demonstrating that on the NBA basketball dataset (as it was at the time), more than $1$ in $20$ tuples appear in the skyline in high dimensions.  Consequently, there have been numerous efforts to derive a smaller cardinality representative subset (e.g.,~\cite{kdomsky,personalizedSky,topkdom}), especially one that presents very distinct tuples (e.g.,~\cite{repskies,dbasedSky}).

Regret and regret minimizing sets are relatively new in the lineage of these efforts.  When introduced by \nanon~\cite{regretMin}, the emphasis was on proving that the maximum regret ratio is bounded by:
$$\frac{d-1}{(\regretsize -d +1)^{d-1}+d-1}.$$ 

Naturally, this bound holds for our generalisation introduced in this paper.  As far as we know, no research has yet concerned itself with finding optimal regret minimizing sets.  

The top-$k$ query off which {\em regret} is based is well studied.  The pareto-dominance graph~\cite{pdg} uses ideas of pareto-optimality to index for top-$k$ queries.  The Onion Technique~\cite{onionTechnique} is a depth-based approach, but suffers the same curse of dimensionality as the skyline.  Ilyas offers a nice survey on the topic of top-$k$ queries~\cite{ilyas}.  Duality transforms are pervasive in this research area (e.g.,~\cite{das,tsaparas}).  We use the duality transform of Chester\ea~\cite{mrtop,vecProj} because of the immediate results on top-$k$ depth contours it provides in answering the reverse top-$k$ queries of Vlachou\ea~\cite{rtopVlachou,vlachouConf}.

Using duality transforms on data points casts the problem into the context of arrangements.  In two dimensions, plane sweep algorithms~\cite{topSweep} are a typical approach to solving problems on arrangements of lines.  Agarwal\ea~\cite{chains} give bounds on the the number of edges and vertices that can exist at each level (or depth) of an arrangement.  Top-$k$ depth contours are not the only notion of contours or depth in arrangements of lines: Hugg\ea~\cite{tufts} evaluate several depth measures and Zuo\ea~\cite{zuo} derive general stastical results that apply to many of them and could be useful in extending this work.  Rousseuw and Hubert consider depth in arrangements for dimensions greater than two~\cite{arrangeDepth}, which could present deeper insight into the greey algorithm presented here.

\section{Conclusion}\label{sec:conclusion}
Regret minimizing sets are a nice alternative to skyline as a succinct representative subset of a dataset, but suffer from a very strict assumption that users expect to see their top-$1$ choice for their queries.  We generalised the notion to that of {\em $k$-regret minimizing sets}, which evaluates how representative a subset of a dataset is not by how closely it approximates every users' top-$1$ choice, but their top-$k$ choice.  We showed that in dual space, the top-$k$ depth contour of a dataset is exactly the optimal $k$-regret minimizing set.  If the cardinality of the $k$-regret minimizing set is specified as an input parameter, then the convex chain that minimizes the ratio of distances from itself to the contour and the contour to the origin is precisely the optimal solution.  We used these ideas to construct an $\order(n^2\regretsize)$ optimal algorithm for two dimensions and a greedy algorithm for general dimension.   

%A couple interesting directions for future work emerge from this.  Firstly, for two dimensions, both the algorithm of Chester\ea~\cite{mrtop} to compute the top-$k$ depth contour and our algorithm here to compute the optimal $k$-regret minimizing set are plane sweep algorithms.  So, it seems highly plausible that the two algorithms could be merged rather than executed in succession.  Regarding general dimension, the convergence rate of our greedy algorithm is an open question.  An empirical investigation could reveal how quickly the algorithm progresses towards the optimal solution and whether effective search space pruning techniques exist.

\balance

\bibliographystyle{abbrv}

\newpage
\begin{appendix}
\section{Algorithm Pseudocode}
\nobalance

\begin{algorithm}[htbf]
 \caption{\label{alg:2d}Calculating an optimal $k$-max-regret minimizing set $\ss$ from $\ds\subseteq\mathbb{R}^2$ with $|\ss|\leq \regretsize$}
 \begin{algorithmic}[1]
  \STATE \textbf{Input}: $\contour_k$; $\regretsize$; $\ls$, sorted by $y$-intercept
  \STATE \textbf{Output}: $\ss\subseteq\ls$, the lines that together form an optimal solution $\ss$ with $|\ss|\leq \regretsize$
  \IF{$|\contour_k|\leq \regretsize$}
   \STATE Return $\contour_k$
  \ENDIF
  \STATE Initialize $\queue$ as an empty priority queue; priority is angle of points, desc
  \FORALL{$l\in\ls$}
   \STATE Set $\as(l) = k*[(y\mathrm{-intercept, max}(y\mathrm{-intercept} - y\mathrm{-intercept of}\ \contour_k,0)]$
   \STATE Add to $\queue$ intersect($l$ and next $l$) if not last $l$
  \ENDFOR
  \WHILE{$\queue$ is not empty}
  \STATE Let $p$ be next point in $\queue$
  \STATE Let $\disttock$ be distance ratio of $p$ to $\contour_k$ and $\contour_k$ to $\origin$
   \IF{$p\in\contour_k$}
    \FORALL{($v\in\as$)}
     \STATE Let $\disttock'$ be distance of ratio of line to $\contour_k$ and $\contour_k$ to $\origin$
     \STATE Let $v$ be $max(v, \disttock')$
    \ENDFOR
   \ENDIF
   \STATE Retrieve adjacent $l_{i}, l_{i+1}$ that intersect at $p$
   \STATE Add intersect($l_{i-1}, l_{i+1}$) if angle less than that of $p$
   \STATE Add intersect($l_{i}, l_{i+2}$) if angle less than that of $p$
   \FORALL{$j\in[0,\regretsize)$}
    \STATE Let $\as(l_{i})_j = \max(\as(l_{i})_j, \disttock)$
    \IF{$j\geq 0$}
     \IF{$\as(l_{i+1})_{j-1} < \as(l_{i+1})_j$}
      \STATE Add $p$ to path of $\as(l_{i+1})_j$
      \STATE Let $\as(l_{i+1})_j = \max(\as(l_{i+1})_{j-1}, \disttock)$
     \ELSE
      \STATE Let $\as(l_{i+1})_j = \max(\as(l_{i+1})_j, \disttock)$
     \ENDIF
    \ELSE
     \STATE Let $\as(l_{i+1})_j = \max(\as(l_{i+1})_j, \disttock)$
    \ENDIF
   \ENDFOR
   \STATE Swap $l_{i}$ and $l_{i+1}$
  \ENDWHILE
  \FORALL{$l\in\ls, j\in[0,\regretsize)$}
   \STATE Remember $v=\as(l)_j$ if smallest yet seen, breaking ties with smaller $j$
  \ENDFOR
  \STATE RETURN set of lines generating path corresponding to $v$
 \end{algorithmic}
\end{algorithm}

\begin{algorithm}[h]
 \caption{\label{alg:allD}Greedy algorithm to compute $k$-max-regret minimizing set $\ss$ from $\ds\subseteq\mathbb{R}^d$ with $|\ss|\leq m$}
 \begin{algorithmic}[1]
  \STATE \textbf{Input}: $\contour_k$; $m$; $\ls$
  \STATE \textbf{Output}: $\ss\subseteq\ls$, the lines that together form a solution $\ss$ with $|\ss|\leq m$
  \IF{$|\contour_k|\leq m$}
   \STATE Return $\contour_k$
  \ENDIF
  \STATE Select an arbitrary set $\ss\subseteq\ls$, such that $|\ss|=m$
  \STATE Let $p$ be point at which distance ratio from lower envelope of $\ss$ to $\contour_k$ is maximized
  \STATE Place all lines $l_i\not\in\ss$ into unsorted queue, $\queue$.
  \WHILE{$\queue$ is not empty}
   \STATE Let $l$ be next line in $\queue$
   \IF{$l$ intersects $[\origin,p]$}
    \FORALL{($l'\in\ss$)}
     \IF{Distance ratio of $\ss\setminus\{l'\}\cup\{l\}$ to $\contour_k <$ distance ratio of $\ss$ to $\contour_k$}
      \STATE Let $\ss$ be $\ss\setminus\{l'\}\cup\{l\}$
      \STATE Restore to $\queue$ all $l_i\not\in\ss$
      \STATE Break
     \ENDIF
    \ENDFOR
   \ENDIF
  \ENDWHILE
  \STATE RETURN $\ss$
 \end{algorithmic}
\end{algorithm}

\end{appendix}

\end{document}